\documentclass{aic}

\usepackage[absolute]{textpos}

\usepackage{thmtools}
\usepackage{thm-restate}

\usepackage{algpseudocode}

\usepackage{hyperref}

\usepackage{xfrac}

\declaretheorem[name=Theorem, numberwithin=section]{theorem}
\declaretheorem[name=Lemma, sibling=theorem]{lemma}

\declaretheorem[name=Definition, sibling=theorem]{definition}
\declaretheorem[name=Corollary, sibling=theorem]{corollary}
\declaretheorem[name=Conjecture, sibling=theorem]{conjecture}
\declaretheorem[name=Claim]{claim}

\newcommand{\Oh}{{\cal O}}

\newcommand{\Cc}{{\cal C}}
\newcommand{\Dd}{{\cal D}}

\newcommand{\Hh}{{\cal H}}
\newcommand{\Tt}{{\cal T}}
\newcommand{\Ff}{{\cal F}}
\newcommand{\Pp}{{\cal P}}
\newcommand{\Qq}{{\cal Q}}
\newcommand{\Rr}{{\cal R}}

\newcommand{\Lb}{\overline{\Lambda}}

\newcommand{\N}{{\mathbb N}}

\newcommand{\rename}{\mathsf{rename}}
\newcommand{\join}{\mathsf{join}}
\newcommand{\union}{\mathsf{union}}
\newcommand{\Lev}{\mathsf{Level}}
\newcommand{\Comp}{\mathsf{Complexity}}
\newcommand{\tp}{\mathsf{top}}
\newcommand{\clw}{\diamondsuit}

\newcommand{\ang}[2]{#1\langle #2\rangle}

\renewcommand{\setminus}{-}

\renewcommand{\leq}{\leqslant}
\renewcommand{\geq}{\geqslant}

\def\cqedsymbol{\ifmmode$\lrcorner$\else{\unskip\nobreak\hfil
\penalty50\hskip1em\null\nobreak\hfil$\lrcorner$
\parfillskip=0pt\finalhyphendemerits=0\endgraf}\fi} 

\newcommand{\cqed}{\renewcommand{\qed}{\cqedsymbol}}

\aicAUTHORdetails{%
  title = {Graphs of Bounded Cliquewidth\\ are Polynomially $\chi$-bounded}, 
  author = {Marthe Bonamy and Micha\l \ Pilipczuk},
  plaintextauthor = {Marthe Bonamy, Michal Pilipczuk},
  plaintexttitle = {Graphs of Bounded Cliquewidth\\ are Polynomially $\chi$-bounded}, 
  runningtitle = {Graphs of Bounded Cliquewidth are Polynomially $\chi$-bounded}, 
}   

\aicEDITORdetails{%
   year={2020},
   number={8},
   received={1 January 2020},   
   published={10 July 2020},  
   doi={10.19086/aic.13668},      
}   

\begin{document}

\begin{frontmatter}[classification=text]

\title{Graphs of bounded cliquewidth\\ are polynomially $\chi$-bounded} 

\author[mb]{Marthe Bonamy\thanks{Supported by the ANR project GrR ANR-18-CE40-0032}}
\author[mp]{Micha\l \ Pilipczuk\thanks{Supported by the project TOTAL that has received funding from the European Research Council (ERC) under the European Union's Horizon 2020 research and innovation programme (grant agreement No.~677651)}}

\begin{abstract}
We prove that if $\Cc$ is a hereditary class of graphs that is polynomially $\chi$-bounded, then the class of graphs that admit decompositions into pieces belonging to $\Cc$ along cuts of bounded rank is also polynomially $\chi$-bounded. In particular, this implies that for every positive integer $k$, the class of graphs of cliquewidth at most $k$ is polynomially $\chi$-bounded.
\end{abstract}
\end{frontmatter}

\begin{textblock}{20}(0, 13.3)
\includegraphics[width=40px]{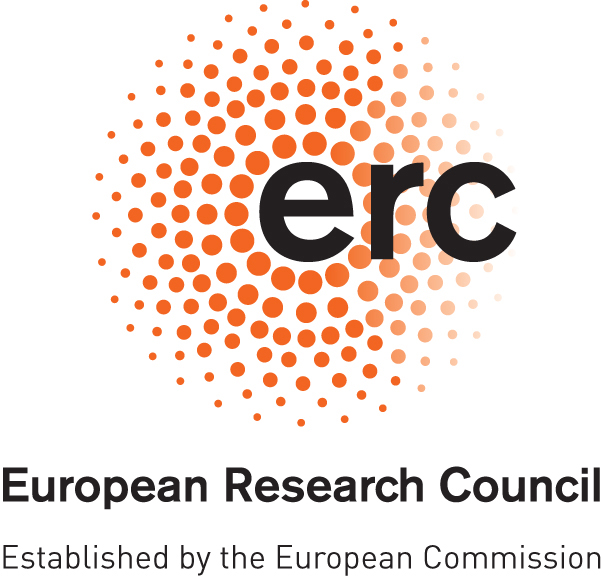}%
\end{textblock}
\begin{textblock}{20}(-0.25, 13.7)
\includegraphics[width=60px]{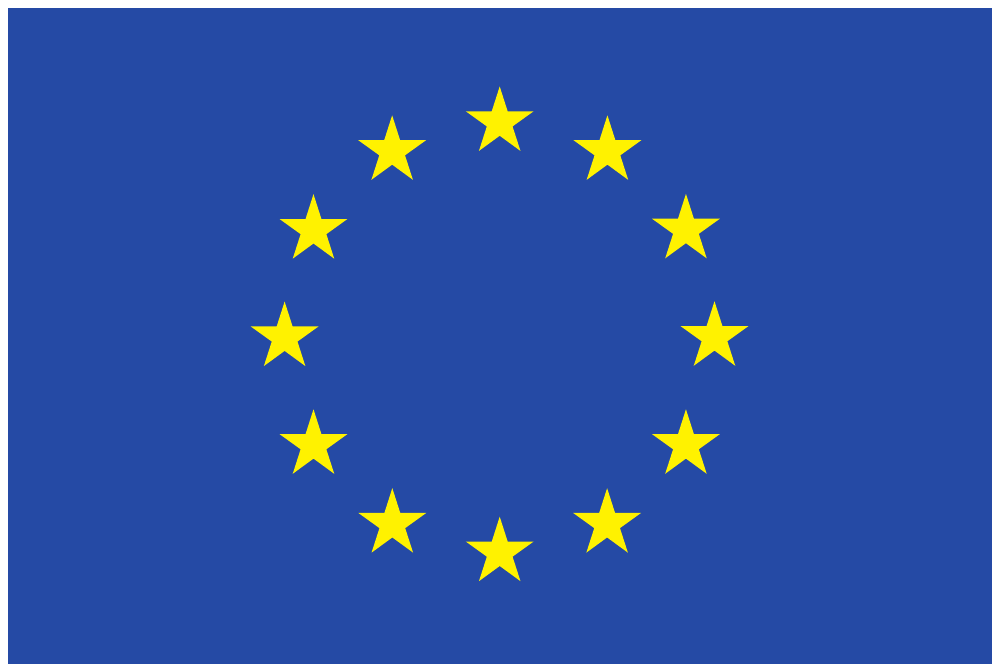}%
\end{textblock}

\section{Introduction}\label{sec:intro}

\paragraph*{$\chi$-boundedness.}
A class of graphs $\Cc$ is {\em{$\chi$-bounded}} if there exists a function $f\colon \N\to \N$, called the {\em{$\chi$-bounding function}}, such that $\chi(G)\leq f(\omega(G))$ for every graph $G\in \Cc$. Here, $\chi(G)$ is the {\em{chromatic number}}  of $G$ --- the least number of colors needed for a proper coloring of $G$ --- and $\omega(G)$ is the {\em{clique number}} of $G$ --- the size of a largest set of pairwise adjacent vertices in $G$. 

The notion of $\chi$-boundedness was proposed by Gy\'arf\'as in~\cite{Gyarfas87} and is a natural generalization of the concept of perfect graphs. Since then, many classes of not necessarily perfect graphs have been shown to be $\chi$-bounded: examples include circular arc graphs~\cite{Gyarfas87},  intersection graphs of axis-parallel boxes in $d$-dimensional space~\cite{Gyarfas87}, circle graphs~\cite{Gyarfas85,Gyarfas85corr}, odd-hole-free graphs~\cite{ScottS16}, long-hole-free graphs~\cite{ChudnovskySS17}, and graphs excluding an induced subdivision of a fixed tree~\cite{Scott97}. Many of those classes have natural interpretations as intersection graphs of geometric objects, but $\chi$-boundedness is not an ubiquitous phenomenon in the geometric setting; for instance, the class of intersection graphs of segments in the plane is not $\chi$-bounded~\cite{PawlikKKLMTW14}. We refer the reader to the recent survey of Scott and Seymour~\cite{ScottS18} for a broader introduction to the topic.

While in the topic of $\chi$-boundedness many basic questions still remain open, even less is known about the optimum asymptotics of $\chi$-bounding functions. Here, we say that a class is {\em{polynomially $\chi$-bounded}} if the $\chi$-bounding function can be chosen to be a polynomial; we can define e.g. linear or quadratic $\chi$-boundedness in the same way. Curiously, we do not know any hereditary graph class that would be $\chi$-bounded, but (provably) not polynomially $\chi$-bounded~\cite{ciacho-thesis}. For several concrete hereditary graph classes one can establish polynomial $\chi$-boundedness, see the recent survey of Schiermeyer and Randerath~\cite{SchiermeyerR19}. However, in many important cases, including the advances on variants of the Gy\'arf\'as-Sumner conjecture and other conjectures of Gy\'arf\'as~\cite{ChudnovskySS17,Scott97,ScottS16}, the known proofs lead only to exponential upper bounds. Also, while the $\chi$-boundedness of circle graphs was known since 80s~\cite{Gyarfas85,Gyarfas85corr}, their polynomial $\chi$-boundedness was established only very recently~\cite{DaviesMcC19}. 

\paragraph*{Cliquewidth.}
The primary object of interest in this work is the class of graphs of cliquewidth at most $k$, for any fixed $k\in \N$. The {\em{cliquewidth}} of a graph $G$ is a parameter that measures the complexity of $G$ in terms of the number of labels needed to construct $G$ by means of certain algebraic operations. Roughly equivalently (up to a multiplicative factor of $2$), one can imagine that a graph of cliquewidth at most $k$ can be hierarchically decomposed into smaller and smaller pieces, up to single vertices, so that in each step of the decomposition we partition the current vertex set into two subsets so that vertices on each side of the partition have only at most $k$ different neighborhoods on the other side. This view is closely related to the notion of {\em{rankwidth}} of $G$, where we measure the complexity of the partition as above not in terms of its {\em{diversity}} --- the number of equivalence classes of the relation of having the same neighborhood on the other side --- but in terms of its {\em{cutrank}} --- the rank over $\mathbb{F}_2$ of the adjacency matrix between the sides. Indeed, it is known that the cliquewidth of a graph of rankwidth $r$ is between $r$ and $2^{r+1}-1$~\cite{OumS06}. Hence, a graph class has bounded rankwidth (i.e., the supremum of the rankwidth of its members is finite) if and only if it has bounded cliquewidth.

Dvo\v{r}\'ak and Kr\'al' proved in~\cite{DvorakK12} that for every $k\in \N$, the class of graphs of cliquewidth at most $k$ is $\chi$-bounded\footnote{Note that Dvo\v{r}\'ak and Kr\'al' state their results in the setting of rankwidth instead of cliquewidth, but, as explained above, the boundedness of these two parameters is equivalent.}. While the obtained $\chi$-bounding function is not stated explicitly in~\cite{DvorakK12}, a careful examination of the proof shows that it is exponential in $\omega(G)$, even for constant $k$. 

In fact, Dvo\v{r}\'ak and Kr\'al' in~\cite{DvorakK12} prove a stronger claim, which can be informally stated as follows: if a graph class $\Cc$ is $\chi$-bounded, then for every $k\in \N$, the class of graphs that admit decompositions into pieces belonging to $\Cc$ along cuts of rank at most $k$ is also $\chi$-bounded. While the result for graphs of bounded cliquewidth follows by applying this statement for a trivial class $\Cc$, one can also obtain other corollaries. For instance, by composing their main result with a decomposition theorem of Geelen for graphs excluding the wheel $W_5$ as a vertex minor~\cite{geelen-thesis},  Dvo\v{r}\'ak and Kr\'al' conclude that this graph class is $\chi$-bounded. Here, the wheel $W_5$ consists of a cycle on $5$ vertices plus one additional vertex connected to all of them.

\paragraph*{Our results.} 
In this work we prove that every graph class of bounded cliquewidth is not only $\chi$-bounded, but in fact polynomially $\chi$-bounded.

\begin{theorem}\label{thm:chibounded}
For every $k \in \N$, the class of graphs of cliquewidth at most $k$ is polynomially $\chi$-bounded.
\end{theorem}

Let us note that in Theorem~\ref{thm:chibounded}, the degree of the $\chi$-bounding polynomial for graphs of cliquewidth $k$ grows with $k$. As we explain in Section~\ref{sec:asymptotics} (see Lemma~\ref{lem:grows-with-k}), this is unavoidable. 

We also remark that independently of us, Ne\v{s}et\v{r}il et al. proved that for every fixed $k\in \N$, graphs of {\em{linear cliquewidth}} at most $k$ are linearly $\chi$-bounded~\cite{NesetrilORS19}. In their work, the result comes as a by-product in the proof of a statement about logical interpretability of graphs of bounded linear cliquewidth with no large half-graphs in graphs of bounded pathwidth. They provide two proofs of their main claims, one using the same tool as we do --- Simon's factorization --- and a second more direct and involved, but yielding better asymptotic bounds. In fact, a close inspection shows that their first proof restricted to arguing only linear $\chi$-boundedness is essentially equivalent to our proof restricted to classes of bounded linear cliquewidth, with the exception that a suitable amortization argument (disguised as applying the fact that cographs are perfect) is used to argue that the $\chi$-bounding function is linear, instead of just polynomial.


In fact, our work applies not only to classes of bounded cliquewidth. We prove a more general claim analogous to that of Dvo\v{r}\'ak and Kr\'al': if $\Cc$ is a polynomially $\chi$-bounded graph class and $k\in \N$, then the class of graphs that can be decomposed into pieces from $\Cc$ along cuts of diversity at most $k$ is also polynomially $\chi$-bounded. A formal statement of this result requires some technical definitions and will be presented in Section~\ref{sec:prelims}, see Theorem~\ref{thm:main-technical} there. 

It seems that our main result may be applicable for establishing polynomial $\chi$-boundedness of classes defined by forbidding fixed vertex-minors, similarly to the work of Dvo\v{r}\'ak and Kr\'al'~\cite{DvorakK12}; we discuss these connections in Section~\ref{sec:wheel-free}. In particular, we generalize two related statements that were used in this context: the result of Chudnovsky et al.~\cite{ChudnovskyPST13} that the closure of a polynomially $\chi$-bounded class under the {\em{substitution}} operation is also polynomially $\chi$-bounded, and the recent result of Kim et al.~\cite{KiKOS19} that the same holds also for the operation of taking {\em{$1$-joins}}. Indeed, these two cases follow from taking $k=1$ and $k=2$ in our main theorem, respectively. In fact, in our proof we rely on the result of Chudnovsky et al.~\cite{ChudnovskyPST13}.

\paragraph*{Our techniques.}
The first ingredient of the proof of our main result is the deterministic variant of Simon's Factorization Forest Theorem due to Colcombet~\cite{Colcombet07}, which provides a factorization theorem for trees labelled with elements of a finite semigroup. The idea is that we can view a decomposition of a graph using cuts of diversity at most $k$ as a tree labelled with elements of the semigroup of relabelings (functions) on a set of $k$ labels. Thus, by applying Colcombet's result we can hierarchically factorize every such decomposition so that the following properties hold:
\begin{itemize}
    \item the factorization has ``depth'' $2^{\Oh(k\log k)}$; and
    \item in every step we partition the current decomposition into factors (which will be factorized further in the following steps) so that the overall structure of factors is a tree that either has depth at most two, or satisfies a certain Ramsey-type condition.
\end{itemize} 
Thus, we can reduce the original statement of the main result to $2^{\Oh(k\log k)}$ applications of a weaker statement, where the provided decomposition of the graph either has depth at most two or enjoys strong Ramsey properties. Decompositions of the former type are called {\em{shallow}}, and decompositions of the latter type are called {\em{splendid}}.

While the case of decompositions of depth two follows from a straightforward product argument, the case of splendid decompositions requires a non-trivial reasoning. For this, we use the aforementioned result of Chudnovsky et al.~\cite{ChudnovskyPST13} about closures of polynomially $\chi$-bounded classes under the substitution operation. The splendidness of the decomposition implies that if we partition edges of the graph into those ``introduced'' at odd and even levels, in both cases we observe a subgraph that can be obtained from graphs from the base class $\Cc$ by a repeated application of substitutions. Hence, we can apply the result of Chudnovsky et al.~\cite{ChudnovskyPST13} to both these subgraphs and conclude by taking the product of the obtained colorings.

\section{Preliminaries}\label{sec:prelims}

For a positive integer $p$, we write $[p]=\{1,\ldots,p\}$.

\paragraph*{Graph terminology.} In this work we consider only finite, undirected graphs. For a graph $G$, the vertex and edge sets of $G$ are denoted by $V(G)$ and $E(G)$, respectively. The {\em{neighborhood}} of a vertex $u$ consists of all vertices adjacent to $u$ and is denoted by $N^G(u)$. For a graph~$G$, we write $\omega(G)$ for the size of the largest {\em{clique}} in $G$, i.e., a set of pairwise adjacent vertices. A {\em{coloring}} of $G$ is any function that maps vertices of $G$ to some set of colors, and it is {\em{proper}} if for every edge $e$ of $G$, the endpoints of $e$ receive different colors in the coloring. For a graph~$G$, the {\em{chromatic number}} of $G$, denoted $\chi(G)$, is the least number of colors needed for a proper coloring of $G$.

A {\em{graph class}} is just a set of graphs, usually infinite. A class of graphs $\Cc$ is {\em{hereditary}} if it is closed under taking induced subgraphs, that is, if $G\in \Cc$ then every graph obtained from $G$ by deleting vertices also belongs to $\Cc$. A graph class $\Cc$ is {\em{$\chi$-bounded}} if there exists a function $f\colon \N\to \N$ such that $\chi(G)\leq f(\omega(G))$ for every $G\in \Cc$. If the function $f$ can be chosen to be a polynomial, then we say that $\Cc$ is {\em{polynomially $\chi$-bounded}}.

A {\em{rooted tree}} is a connected graph without cycles with one vertex designated to be the root. As all decompositions in this work will have a form of rooted trees, for distinguishment we usually use the term {\em{node}} for a vertex of a tree. If $T$ is a rooted tree with nodes $x$ and $y$, then we say that $x$ is an {\em{ancestor}} of $y$, equivalently that $y$ is a {\em{descendant}} of $x$, if $x$ lies on the unique path in $T$ from $y$ to the root. Note that every node is considered its own ancestor and descendant. If additionally $x$ and $y$ are adjacent, we say that $x$ is a {\em{parent}} of $y$ and equivalently that $y$ is a {\em{child}} of $x$. Note that while every node has a unique parent --except for the root which has none-- a node may have arbitrarily many children.

\paragraph*{Cliquewidth.} A {\em{$k$-labelled}} graph is a graph $G$ together with a labelling $\lambda\colon V(G)\to [k]$. On $k$-labelled graphs we define the following operations:
\begin{itemize}
    \item for $i,j\in [k]$, $i\neq j$, the operation $\join_{i,j}(\cdot)$ adds all possible edges with one endpoint of label~$i$ and the other of label~$j$;
    \item for $i,j\in [k]$, $i\neq j$, the operation $\rename_{i\to j}(\cdot)$ changes the label of each vertex labelled $i$ to~$j$;
    \item the operation $\union(\cdot,\cdot)$ outputs the disjoint union of two $k$-labelled graphs.
\end{itemize}
Note that $\join_{i,j}(\cdot)$ and $\rename_{i\to j}(\cdot)$ are unary operations, that is, they are applied to a single $k$-labelled graph, while $\union(\cdot,\cdot)$ is a binary operation. The {\em{cliquewidth}} of an (unlabelled) graph $G$ is the least integer $k$ such that some $k$-labelling of $G$ can be constructed using the operations described above from single-vertex $k$-labelled graphs.

Note that the construction of a graph $G$ of cliquewidth $k$, as described above, naturally gives rise to a term over an algebra of operations $\join_{i,j}(\cdot)$, $\rename_{i\to j}(\cdot)$, and $\union(\cdot,\cdot)$ whose leaves are single-vertex $k$-labelled graphs. This term is called a {\em{$k$-expression}} that constructs $G$. 

\paragraph*{Decompositions.}
As outlined in Section~\ref{sec:intro}, we will actually prove a result stronger than claimed in Theorem~\ref{thm:chibounded}: whenever $\Cc$ is a hereditary class of graphs that is polynomially $\chi$-bounded, the class of graphs that admit decompositions using cuts of bounded diversity into pieces that belong to $\Cc$ is also polynomially $\chi$-bounded. To state this result formally, we need to introduce the notion of a decomposition that will be used. It is actually a rooted variant of {\em{${\cal G}$-bounded decompositions}} used by Dvo\v{r}\'ak and Kr\'al'~\cite{DvorakK12}, hence we use a similar terminology, but adjusted to the rooted setting.

\begin{definition}
A {\em{decomposition}} of a graph $G$ is a rooted tree $T$ together with a function $\eta$ that maps vertices of $G$ to nodes of $T$. 
\end{definition}

Note in the above definition, there are no requirements about surjectivity or injectivity of the mapping $\eta$. In particular, many vertices of $G$ may be mapped to the same node of $T$. Also, there are no restrictions on the number of children of any node in $T$.

Whenever a decomposition $\Tt=(T,\eta)$ of a graph $G=(V,E)$ is clear from the context, we use the following notation. For two nodes $x,y$ of $T$, we write $x\preceq y$ if $x$ is an ancestor of $y$ in $T$; recall that we consider every node to be its own ancestor as well. For an edge $e=uv$ of $G$, we write
$$\eta(e)=\textrm{the least common ancestor of }\eta(u)\textrm{ and }\eta(v)\textrm{ in }T.$$
For a node $x$ of $T$, we write
$$\ang{V}{x}=\{v\in V\colon \eta(v)\succeq x\},\qquad \ang{E}{x}=\{e\in E\colon \eta(e)=x\},\qquad \ang{G}{x}=(\ang{V}{x},\ang{E}{x}).$$
If the decomposition $\Tt$ for which the above objects are defined is not clear from the context, we write it as the second argument in the brackets, i.e. we write $\ang{V}{x,\Tt},\ang{E}{x,\Tt},\ang{G}{x,\Tt}$.

Let us remark that we believe that the graph $\ang{G}{x}$, as defined above, may be the right analogue of the {\em{torso}} of a node of $x$ from the context of classic tree decompositions. Note that for every child $y$ of $x$, the vertices of $\ang{V}{y}$ form an independent set in $\ang{G}{x}$.

\begin{definition}
For a class of graphs $\Cc$, a decomposition $(T,\eta)$ of a graph $G$ is {\em{$\Cc$-governed}} if $\ang{G}{x}\in \Cc$ for every node $x$ of $T$.
\end{definition}

For a node $x$ of $T$, we introduce the following equivalence relation $\sim_x$ on $\ang{V}{x}$:
$$u\sim_x v\quad \Leftrightarrow \quad N^G(u)\setminus \ang{V}{x} = N^G(v)\setminus \ang{V}{x}.$$
In other words, $u,v\in \ang{V}{x}$ are considered equivalent if they have exactly the same neighborhood outside of $\ang{V}{x}$. Note that if $y$ is a descendant of $x$, then $u\sim_y v$ entails $u\sim_x v$.

\begin{definition}
Let $(T,\eta)$ be a decomposition of a graph $G$.
The {\em{diversity}} of a node $x$ of $T$ is the number of equivalence classes of $\sim_x$. The {\em{diversity}} of the decomposition $(T,\eta)$ is the largest diversity among the nodes of $T$.
\end{definition}

Finally, we define the depth of a decomposition.

\begin{definition}
Let $(T,\eta)$ be a decomposition of a graph $G$. The {\em{depth}} of $(T,\eta)$ is the maximum number of edges in a path from the root to a node in $T$. A decomposition is {\em{shallow}} if it has depth at most two.
\end{definition} 

\paragraph*{Statement of the results.}
With all the terminology prepared, we can formally state the main result of this paper.

\begin{theorem}\label{thm:main-technical}
Let $\Cc$ be a hereditary graph class that is polynomially $\chi$-bounded. Then for every fixed $k\in \N$, the class of graphs that admit $\Cc$-governed decompositions of diversity at most $k$ is also polynomially $\chi$-bounded.
\end{theorem}

Now, Theorem~\ref{thm:chibounded} is directly implied by Theorem~\ref{thm:main-technical} and the next lemma, which follows almost immediately from the definition of cliquewidth.

\begin{lemma}\label{lem:cwk}
For every $k\in \N$, every graph of cliquewidth at most $k$ admits a decomposition of diversity at most $k$ governed by the class of bipartite graphs. 
\end{lemma}
\begin{proof}
Let $G$ be a graph of cliquewidth at most $k$ and let $\tau$ be a $k$-expression constructing (some $k$-labelling of) $G$. 
Let $T$ be the rooted tree of the $k$-expression $\tau$, i.e., the tree of the subterms of $\tau$ with the descendant relation defined by the subterm relation. Note that the vertices of $G$ are in one-to-one correspondence with the leaves of $\tau$ where they are introduced, hence let $\eta\colon V(G)\to V(T)$ map every vertex of $G$ to the corresponding leaf of $T$. Then $(T,\eta)$ is a decomposition of $G$. 

Observe that for every node of $x$, say corresponding to a subterm $\tau_x$ of $\tau$, the set of vertices introduced in the descendants of $x$ is exactly $\ang{V}{x}$. By the definition of a $k$-expression, all vertices assigned the same label in the graph constructed by $\tau_x$ have the same neighborhood outside of $\ang{V}{x}$ in $G$, hence $\sim_x$ has at most $k$ equivalence classes. It follows that $(T,\eta)$ has diversity at most $k$.

To see that $(T,\eta)$ is governed by the class of bipartite graphs, observe that every internal node $x$ has either one child (if it corresponds to a $\rename$ operation) or two children (if it corresponds to a $\join$ operation) and in both cases $\eta^{-1}(x)=\emptyset$, hence $\ang{G}{x}$ is bipartite. On the other hand, if $x$ is a leaf then $\ang{G}{x}$ has one vertex, hence it is bipartite as well.
\end{proof}


As the proof of Theorem~\ref{thm:main-technical} spans the whole remainder of the paper, from now on we fix the integer $k\in \N$ and the hereditary graph class $\Cc$. Also, we write $\Cc^{\clw}$ for the latter class considered in Theorem~\ref{thm:main-technical}, i.e. the class of graphs that admit $\Cc$-governed decompositions of diversity at most $k$.

\section{Stratifying the class $\Cc^{\clw}$}

The first step in the proof of Theorem~\ref{thm:main-technical} is to stratify the inclusion $\Cc\subseteq \Cc^{\clw}$. That is, we introduce a sequence of classes $\Cc=\Dd_0\subseteq \Dd_1\subseteq \ldots\subseteq \Dd_p=\Cc^\clw$, for some $p\in \N$ depending only on $k$, with the following property: for every $i\in \{1,\ldots,p\}$, every graph $G\in \Dd_i$ admits a $\Dd_{i-1}$-governed decomposition of diversity at most $k$ that is moreover somehow well-behaved, to be defined in a moment. Then Theorem~\ref{thm:main-technical} can be proved by an induction over the sequence $\Dd_1,\ldots,\Dd_p$, where every step of the induction boils down to applying Theorem~\ref{thm:main-technical} under the additional supposition of well-behavedness of the decomposition. Thus, we reduce Theorem~\ref{thm:main-technical} to a weaker statement where we have an additional assumption about the decomposition.

To formally explain ``well-behavedness'' we need some additional terminology. We let $\Ff$ be the semigroup of all functions from $[k]$ to $[k]$, with composition being the semigroup action. That is, for $f,g\in \Ff$ we write $f\cdot g\in \Ff$ for the function that maps every $i\in [k]$ to $f(g(i))$.

\begin{definition}
Suppose $(T,\eta)$ is a decomposition of a graph $G$ with diversity at most $k$.
A {\em{tagging}} of $(T,\eta)$ is a family $\Lambda=(\lambda^x)_{x\in V(T)}$ such that every $\lambda_x$ is a mapping from $\ang{V}{x}$ to $[k]$ with the following two properties: $(i)$ whenever $\lambda^x(u)=\lambda^x(v)$ for some $u,v\in \ang{V}{x}$, then $u\sim_x v$, and $(ii)$ whenever $\lambda^x(u)=\lambda^x(v)$ for some $u,v\in \ang{V}{x}$ then $\lambda^y(u)=\lambda^y(v)$ for any ancestor $y$ of $x$. 

With a tagging $\Lambda$ of a decomposition $(T,\eta)$ of $G$ we associate a {\em{labelling}} $\Lb\colon E(T)\to \Ff$ defined as follows: for every edge $e$ of $T$, say $e=xy$ where $x$ is the parent of $y$, we set $\Lb(e)\in \Ff$ to be the function that maps every $i\in [k]$ to $\lambda^x(u)$ if there is some $u\in \ang{V}{y}$ satisfying $\lambda^y(u)=i$ and to $1$ otherwise.
\end{definition}

Note that every decomposition of diversity at most $k$ has a tagging: any tagging can be obtained by enumerating the classes of $\sim_x$ with numbers from $1$ to $k$, in any way. We will, however, consider decompositions that admit taggings which give rise to somewhat restricted labellings, as explained next.



\begin{definition}\label{def:forward-Ramsey}
Let $S$ be a semigroup. A set of elements $A\subseteq S$ is {\em{forward Ramsey}} if for all $e,f\in A$ we have $e\cdot f = e$. In particular, each $e\in A$ is an {\em{idempotent}} in $S$, that is, $e\cdot e=e$.
\end{definition}

\begin{definition}
A decomposition $(T,\eta)$ of a graph $G$ of diversity at most $k$ is called {\em{splendid}} if there exists a tagging $\Lambda$ of $(T,\eta)$ such that the set $\{\Lb(e)\colon e\in E(T)\}$ is forward Ramsey in $\Ff$.
\end{definition}

The next lemma formalizes the idea of stratification of the inclusion $\Cc\subseteq \Cc^\clw$ and is the main result of this section.

\begin{lemma}\label{lem:hierarchy}
There exists $p\in 2^{\Oh(k\log k)}$ and a sequence of hereditary graph classes
$$\Cc=\Dd_0\subseteq \Dd_1\subseteq\ldots\subseteq \Dd_p=\Cc^\clw$$
such that for every $i\in [p]$, every graph $G\in \Dd_i$ admits a $\Dd_{i-1}$-governed decomposition of diversity at most $k$ that is either splendid or shallow.
\end{lemma}

Using Lemma~\ref{lem:hierarchy} and straightforward induction, the proof of Theorem~\ref{thm:main-technical} boils down to proving the statement under the additional assumption that the decomposition witnessing $G\in \Cc^\clw$ is either splendid or shallow. We treat this case in Section~\ref{sec:splendid}, while for the rest of this section we concentrate on proving Lemma~\ref{lem:hierarchy}.

\bigskip

Our main tool for the proof of Lemma~\ref{lem:hierarchy} is the deterministic variant of Simon's Factorization Forest Theorem, due to Colcombet~\cite{Colcombet07}. This result originates in the algebraic theory of formal languages, so to state it we need several auxiliary definitions.

\paragraph*{Deterministic Simon's factorization.}
In the following, a {\em{word}} over a semigroup $S$ is a (possibly empty) sequence of elements of $S$. The set of all words over $S$ is denoted by $S^*$, whereas $S^+$ denotes the set of all non-empty words over $S$. For a word $w\in S^+$ of length $n$ and positions $0\leq x<y\leq n$, we write $w[x,y]$ for the subword of $w$ starting with the symbol at position $x+1$ and ending with the symbol at position $y$, where positions are numbered from $1$ to $n$.

The concatenation of two words $u$ and $v$ will be denoted by $uv$. Note that $S^+$ endowed with concatenation is an (infinite) semigroup, with a natural homomorphism $\phi\colon S^+\to S$ defined as follows: for $w\in S^+$, $\phi(w)$ computes the product in $S$ of the symbols appearing in $w$ from left to right. 

Suppose $w$ is a word over $S$, say of length $n$. A {\em{split}} of $w$ of {\em{height}} $h$ is a mapping $s\colon \{0,\ldots,n\}\to [h]$. Two positions $0\leq x\leq y\leq n$ are {\em{$s$-equivalent}}, denoted $x\sim_s y$,~if 
$$s(x)=s(y)\qquad \textrm{and}\qquad s(z)\leq s(x)\quad \textrm{for all }x\leq z \leq y.$$
A split $s$ of $w$ is called {\em{forward Ramsey}} if the following condition is satisfied: for all quadruples of positions $x,y,x,y'\in \{0,1,\ldots,n\}$ that are pairwise $s$-equivalent and satisfy $x<y$ and $x'<y'$, we have 
$$\phi(w[x,y]) = \phi(w[x,y])\cdot \phi(w[x',y']).$$
With this terminology introduced, we can state the result of Colcombet.

\begin{theorem}[Theorem~1 of~\cite{Colcombet07}]\label{thm:colcombet}
For every finite semigroup $S$ there exists a mapping $\mu\colon S^*\to [|S|]$ satisfying the following property.
For a word $w\in S^+$, say of length $n$, define a split $s_w\colon \{0,1,\ldots,n\}\to [|S|]$ of $w$ as follows: for $x\in \{0,1,\ldots,n\}$, set $s_w(x)=\mu(w[0,x])$. Then for every $w\in S^+$, the split $s_w$ is forward Ramsey.
\end{theorem}

We remark that in~\cite{Colcombet07} a stronger conclusion is provided, namely that the value of $\mu(w)$ can be computed by running a deterministic automaton with at most $|S|^{|S|}$ states over $w$. We will not use this property here.

\paragraph*{Factorizing trees.}
As observed in~\cite{Colcombet07}, Theorem~\ref{thm:colcombet} is well-suited for the treatment of trees. Essentially, it proves that for any rooted tree whose edges are labelled with elements of a finite semigroup, we can find a bounded-depth hierarchy of forward Ramsey factorizations. This idea is formalized next.

Fix a finite semigroup $S$.
An {\em{$S$-labelled tree}} is a rooted tree $T$ together with a mapping $\rho\colon E(T)\to S$, called further the {\em{labelling}}.
Similarly as before, an $S$-labelled tree $(T,\rho)$ is {\em{splendid}} if the set of labels $\{\rho(e)\colon e\in E(T)\}$ is forward Ramsey in $S$.

A {\em{factorization}} of a rooted tree $T$ is a partition $\Pp$ of $V(T)$ such that every part of $\Pp$ induces a connected subtree of $T$. 
These subtrees are called the {\em{factors}} of $\Pp$ and for a factor $F$ of $\Pp$, we let $\tp(F)$ be the top-most vertex of $F$. Thus, $\tp(\Pp)$ is the set of all top-most vertices of factors of $\Pp$. If $T$ is $S$-labelled, then we consider the factors of $\Pp$ to be $S$-labelled as well by restricting the labelling of $T$.

For a factorization $\Pp$ of a rooted tree $T$, we define the {\em{quotient tree}} $\sfrac{T}{\Pp}$ as follows: the node set of $\sfrac{T}{\Pp}$ is $\tp(\Pp)$ and the ancestor/descendant relation in $\sfrac{T}{\Pp}$ is the restriction of this relation from $T$ to $\tp(\Pp)$. Note that this is equivalent to contracting every factor $F$ of $\Pp$ into a single node named $\tp(F)$, and setting the ancestor/descendant relation between contracted factors in the obvious manner.

If $\Pp$ is a factorization of $T$ and $T$ is $S$-labelled, say with a labelling $\rho$, then we can define a labelling $\sfrac{\rho}{\Pp}$ of $\sfrac{T}{\Pp}$ as follows.
Take any edge $e$ of $\sfrac{T}{\Pp}$, say $e=xy$ where $x$ is the parent of $y$ in $\sfrac{T}{\Pp}$. Then $x$ is an ancestor of $y$ in $T$ and let $P$ be the path in $T$ from $x$ to $y$. If $e_1,e_2,\ldots,e_m$ are the consecutive edges of $T$ on $P$, then we let
$$\sfrac{\rho}{\Pp}(e)=\rho(e_1)\cdot \rho(e_2)\cdot \ldots \cdot\rho(e_m).$$

Now, from Theorem~\ref{thm:colcombet} we can derive the following Simon-like result for trees.

\begin{lemma}\label{lem:colcombet-trees}
For every finite semigroup $S$ there exists a sequence of classes of $S$-labelled trees
$$\Hh_0\subseteq \Hh_1\subseteq \ldots\subseteq \Hh_{3|S|}$$
satisfying the following conditions:
\begin{itemize}
    \item $\Hh_0$ consists only of one tree with one node, while $\Hh_{3|S|}$ is the class of all $S$-labelled trees; and
    \item for every $i\in [3|S|]$, every tree $(T,\rho)\in \Hh_i$ has a factorization $\Pp$ such that all factors of $\Pp$ belong to $\Hh_{i-1}$ and $(\sfrac{T}{\Pp},\sfrac{\rho}{\Pp})$ is either splendid or shallow.
\end{itemize}
\end{lemma}
\begin{proof}
We will use the following notation. For an $S$-labelled tree $(T,\rho)$ and node $x$ of $T$, we write $(T_x,\rho_x)$ for the $S$-labelled tree induced in $(T,\rho)$ by the descendants of $x$ (including $x$ itself); here $\rho_x$ is the restriction of $\rho$ to edges between those descendants. More generally, if $F$ is a subtree of $T$, then $\rho|_F$ is the restriction of $\rho$ to the edges between the nodes of $F$.
Also, we write $P_x$ for the (oriented) path in $T$ that starts at the root of $T$ and ends at $x$. Then $w_x\in S^*$ is the word consisting of elements of $S$ assigned by $\rho$ to consecutive edges of $P_x$. Note that if $x$ is the root of $T$, then $w_x$ is the empty word.

For an $S$-labelled tree $(T,\rho)$ and positive integer $h$, any function $t\colon V(T)\to [h]$ will be called a {\em{split}} of $(T,\rho)$ of {\em{height}} $h$.
Such a split $t$ is {\em{forward Ramsey}} if the following holds: for every node $x$ of $T$, the split $t$ restricted to the nodes on $P_x$ induces a forward Ramsey split of $w_x$. 

We define the {\em{level}} of $(T,\rho)$ as follows:
$$\Lev(T,\rho)=\min\,\{\ h\ \colon\ \textrm{there exists a forward Ramsey split of }(T,\rho)\textrm{ of height }h\ \}.$$
From Theorem~\ref{thm:colcombet} we easily deduce that the $\Lev$ function has bounded range.

\begin{claim}\label{cl:level-bounded}
For every $S$-labelled tree $(T,\rho)$ we have $\Lev(T,\rho)\leq |S|$.
\end{claim}
\begin{proof}
It suffices to construct a split of $(T,\rho)$ of height at most $|S|$.
Let $\mu\colon S^*\to [|S|]$ be the mapping provided by Theorem~\ref{thm:colcombet} for the semigroup $S$. Then let us define a split $t\colon V(T)\to [|S|]$ as follows:
for every $x\in V(T)$ we set $t(x) = \mu(w_x)$.
Then Theorem~\ref{thm:colcombet} directly implies that $t$ is a forward Ramsey split of $(T,\rho)$.
\cqed\end{proof}

On the other hand, we inductively define the notion of the {\em{complexity}} of an $S$-labelled tree $(T,\rho)$ as follows:
\begin{itemize}
    \item If $T$ has one node then $\Comp(T,\rho)=0$.
    \item Otherwise, $\Comp(T,\rho)$ is the least positive integer $i$ such that $(T,\rho)$ admits a factorization $\Pp$ where every factor of $\Pp$ has complexity smaller than $i$ and $(\sfrac{T}{\Pp},\sfrac{\rho}{\Pp})$ is either splendid or shallow.
\end{itemize}
It is not hard to see that every $S$-labelled tree has a finite complexity; say, bounded by its depth. However, this will follow from the following statement, which is the core argument of the proof.


\begin{equation}\tag{$\star$}\label{cl:comp-level}
\textrm{For every $S$-labelled tree $(T,\rho)$ we have $\Comp(T,\rho)\leq 3\cdot \Lev(T,\rho)$.}
\end{equation}

Note that establishing $(\star)$ will conclude the proof, because we will define classes $\Hh_0,$ $\Hh_1, \ldots, \Hh_{3|S|}$ as follows: $\Hh_i$ comprises of all $S$-labelled trees of complexity at most $i$. Then Claim~\ref{cl:level-bounded} and~($\star$) ensure us that $\Hh_{3|S|}$ contains all $S$-labelled trees, while from the definition of the complexity we infer that $\Hh_0$ consists only of the one-node tree and that the second condition from the lemma statement holds.

\medskip

Therefore, from now on we focus on proving ($\star$). For this, we proceed by induction on the level of $(T,\rho)$.

\paragraph*{Base case.} Suppose $(T,\rho)$ has level $1$, which means that it admits a forward Ramsey split $t$ that such that $t(x)=1$ for each node $x$ of $T$.
The following generic claim will be useful.

\begin{claim}\label{cl:sum-fRamsey}
Suppose $A,B\subseteq S$ are forward Ramsey. If $A\cap B\neq \emptyset$, then $A\cup B$ is forward Ramsey as well.
\end{claim}
\begin{proof}
We need to verify that for any $e,f\in A\cup B$, we have $e\cdot f=e$. If $e,f\in A$ or $e,f\in B$ then this follows from the assumption that $A$ and $B$ are forward Ramsey. Hence, by symmetry suppose that $e\in A$ and $f\in B$. Take any $g\in A\cap B$. Since $A$ and $B$ is forward Ramsey, we have
$$e\cdot g = e\qquad \textrm{and}\qquad g\cdot f = g.$$
Therefore,
$$e=e\cdot g = e\cdot (g\cdot f) = (e\cdot g)\cdot f = e\cdot f,$$
as required.
\cqed\end{proof}

Let $r$ be the root of $T$.
For every child $y$ of $r$, let $A_y$ be the set of all elements of $S$ that are assigned by $\rho$ to the edges of the subtree of $T$ induced by the descendants of $y$.

\begin{claim}\label{cl:Ay-fRamsey}
For every child $y$ of $r$, the set $A_y$ is forward Ramsey.
\end{claim}
\begin{proof}
Since the split $t$ is forward Ramsey, for every descendant $z$ of $y$ the set of elements appearing in $w_z$ is forward Ramsey. All these sets for different descendants $z$ contain the element $\rho(ry)$. Hence, by repeated application of Claim~\ref{cl:sum-fRamsey} we conclude that $A_y\cup \{\rho(ry)\}$ is forward Ramsey, hence $A_y$ is forward Ramsey as well.
\cqed\end{proof}

We observe that for every child $y$ of $r$, the tree $(T_y,\rho_y)$ has complexity at most $1$. Indeed, we can take its factorization $\Pp_y$ into single-node factors, which have complexity $0$, and then $(\sfrac{T_y}{\Pp_y},\sfrac{\rho_y}{\Pp_y})=(T_y,\rho_y)$ is splendid due to Claim~\ref{cl:Ay-fRamsey}.
Then $(T,\rho)$ has complexity at most $2$, because we can take its factorization $\Pp$ that puts the root $r$ as a one-node factor and each subtree $T_y$, for $y$ ranging over children of $r$, as a separate factor. Indeed, then $(\sfrac{T}{\Pp},\sfrac{\rho}{\Pp})$ has depth at most $2$ so is shallow. As $\Comp(T,\rho)\leq 2<3$, this concludes the base case of the induction.

\paragraph*{Induction step.} Suppose that $(T,\rho)$ is an $S$-labelled tree of level $\ell>1$. Let $t\colon V(T)\to [\ell]$ be the forward Ramsey split of $(T,\rho)$ that witnesses $\Lev(T,\rho)=\ell$. Further, let $X=t^{-1}(\ell)$. 

Define an equivalence relation $\sim$ on $V(T)$ as follows: $x,y\in V(T)$ are $\sim$-equivalent if either both do not have an ancestor belonging to $X$, or they have the same least ancestor belonging to $X$. It is easy to see that the equivalence classes of $\sim$ induce connected subtrees of~$T$, hence the partitioning $\Pp$ of $V(T)$ according to the equivalence classes of $\sim$ is a factorization of $T$. For each $x\in X$, we let $\Pp_x$ be the factorization of $(T_x,\rho_x)$ consisting only of those factors of $\Pp$ that contain only descendants of $x$.

We first observe that a reasoning similar to that of Claim~\ref{cl:Ay-fRamsey} shows that for ``deep'' elements $y\in X$, $\Pp_y$ is already a good factorization of $(T_y,\rho_y)$.

\begin{claim}\label{cl:shallow-splendid}
Suppose a node $y\in X$ has an ancestor that belongs to $X$ and is different from $y$. Then $(\sfrac{T_y}{\Pp_y},\sfrac{\rho_y}{\Pp_y})$ is splendid.
\end{claim}
\begin{proof}
Let $x$ be the least ancestor of $y$ such that $x\in X$ and $x\neq y$.
Consider any descendant $z$ of $y$ that belongs to $X$ and let $Q_{xz}$ be the path in $\sfrac{T}{\Pp}$ from $x$ to $z$; then $xy$ is the first edge of $Q_{xz}$. Since $t$ is forward Ramsey, $t(x)=t(y)=t(z)=\ell$, and all values of $t$ are not larger than $\ell$, it follows that the set of elements of $S$ assigned by $\sfrac{\rho}{\Pp}$ to the edges of $Q_{xz}$ is forward Ramsey; call this set $B_z$. Observe that sets $B_z$ pairwise share the element $\sfrac{\rho}{\Pp}(xy)$, hence by repeated application of Claim~\ref{cl:sum-fRamsey} we conclude that $\bigcup_{z\succeq y,\,z\in X} B_z$ is forward Ramsey. Since this set contains $\{\sfrac{\rho_y}{\Pp_y}(e)\colon e\in E(\sfrac{T_y}{\Pp_y})\}$, we conclude that $(\sfrac{T_y}{\Pp_y},\sfrac{\rho_y}{\Pp_y})$ is splendid.
\cqed\end{proof}

\begin{claim}\label{cl:Ty-comp}
For every node $y\in X$ that has an ancestor belonging to $X$ but different from $y$, we have $\Comp(T_y,\rho_y)\leq 3\ell-1$.
\end{claim}
\begin{proof}
By Claim~\ref{cl:shallow-splendid}, $(T_y,\rho_y)$ admits factorization $\Pp_y$ such that $(\sfrac{T_y}{\Pp_y},\sfrac{\rho_y}{\Pp_y})$ is splendid. Therefore, it suffices to prove that each factor $F$ of $\Pp_y$ has complexity at most $3\ell-2$. 

Observe that the only node of $F$ mapped to $\ell$ by $t$ is the root of $F$. Therefore, consider the factorization $\Pp_F$ of $F$ that puts the root of $F$ into a one-node factor and every subtree $F'$ of $F$ rooted at a child of the root into a separate factor. Observe that for every such subtree $F'$, restricting $t$ to $F'$ yields a forward Ramsey split of $(F',\rho|_{F'})$ of height at most $\ell-1$, witnessing that $\Lev(F',\rho|_{F'})\leq \ell-1$. By induction hypothesis we have $\Comp(F',\rho|_{F'})\leq 3\ell-3$. As $\sfrac{F}{\Pp_F}$ is shallow and the root factor of $\Pp_F$ has complexity $0$, we conclude that $F$ has complexity at most $3\ell-2$.
\cqed\end{proof}

Let now $R$ be the connected subtree of $T$ induced by those nodes of $T$ that have at most one ancestor belonging to $X$ (including themselves). Note that $R$ contains the root of $T$.

\begin{claim}\label{cl:R-comp}
$\Comp(R,\rho|_R)\leq 3\ell-1$.
\end{claim}
\begin{proof}
First observe that for every $x\in R\cap X$, in the tree $(R_x,\rho|_{R_x})$ --- the subtree of $(T,\rho)$ induced by all descendants of $x$ that belong to $R$ --- the only node mapped to $\ell$ by $\rho$ is the root. Therefore, the same argument as in the proof of Claim~\ref{cl:Ty-comp} shows that $\Comp(R_x,\rho|_{R_x})\leq 3\ell-2$.

Now if the root $r$ of $T$ belongs to $X$, then $(R_r,\rho|_{R_r})=(R,\rho|_R)$, hence $\Comp(R,\rho|_R)\leq 3\ell-2<3\ell-1$, as claimed.
Otherwise, let $R'$ be the connected subtree of $R$ induced by all those nodes of $R$ that have no ancestor in $X$ (including themselves); then $R'$ contains $r$. Note that $t$ restricted to $R'$ is a forward Ramsey split of $(R',\rho|_{R'})$ of height at most $\ell-1$, hence by the induction hypothesis we have $\Comp(R',\rho|_{R'})\leq 3\ell-3$.
Now consider a factorization $\Rr$ of $R$ consisting of factor $R'$ and factors $R_x$ for all $x\in X\cap R$. Since all these factors have complexity at most $3\ell-2$ and $\sfrac{R}{\Rr}$ has depth at most $2$, we conclude that $\Comp(R,\rho|_R)\leq 3\ell-1$, as claimed.
\cqed\end{proof}

Now, construct a factorization $\Qq$ of $(T,\rho)$ as follows: the factors of $\Qq$ consists of $R$ and $T_y$ for each $y\in X$ such that $y$ has exactly one ancestor that belongs to $X$ and is different from $y$.
By Claims~\ref{cl:Ty-comp} and~\ref{cl:R-comp}, each of the factors of $\Qq$ has complexity at most $3\ell-1$. Furthermore, $\sfrac{T}{\Qq}$ has depth at most $2$. This implies that $(T,\rho)$ has complexity at most $3\ell$, which concludes the proof of ($\star$) and of Lemma~\ref{lem:colcombet-trees}.
\end{proof}

With Lemma~\ref{lem:colcombet-trees} established, the conclusion of the proof of Lemma~\ref{lem:hierarchy} boils down to an easy verification.

\setcounter{claim}{0}

\begin{proof}[Proof of Lemma~\ref{lem:hierarchy}]
Apply Lemma~\ref{lem:colcombet-trees} to the semigroup $\Ff$ -- remember that $\Ff$ is the semigroup of all functions from $[k]$ to $[k]$, with composition being the semigroup action. This yields a suitable sequence of classes of $S$-labelled trees $\Hh_0\subseteq \Hh_1\subseteq\ldots \subseteq \Hh_{3|\Ff|}$. We set
$$p=3|\Ff|=3\cdot k^k$$ and define graph classes $\Dd_0\subseteq \Dd_1\subseteq\ldots \subseteq\Dd_p$ as follows: a graph $G\in \Cc^\clw$ belongs to $\Dd_i$ if and only if there exists a $\Cc$-governed decomposition $\Tt=(T,\eta)$ of $G$ of diversity at most $k$ and its tagging $\Lambda$ such that $(T,\Lb)$, treated as an $\Ff$-labelled tree, belongs to $\Hh_i$. We now verify that classes $\Dd_0,\Dd_1\ldots,\Dd_p$ defined in this way satisfy all the required properties.

First, observe that graphs from $\Cc$ are exactly those that admit $\Cc$-governed decompositions with a one-node decomposition tree. Thus $\Dd_0=\Cc$. On the other hand, from Lemma~\ref{lem:colcombet-trees} we conclude that $\Dd_p=\Cc^\clw$. Also, since $\Cc$ is hereditary, so are all the classes $\Dd_0,\Dd_1,\ldots,\Dd_p$.

Now, fix any $i\in \{1,\ldots,p\}$ and consider any graph $G\in \Dd_i$. Let $\Tt=(T,\eta)$ be a decomposition and $\Lambda$ be its tagging that witness the membership $G\in \Dd_i$. Then $(T,\Lb)\in \Hh_i$. By Lemma~\ref{lem:colcombet-trees} we infer that there exists a factorization $\Pp$ of $T$ such that every factor of $\Pp$ belongs $\Hh_{i-1}$ and $(\sfrac{T}{\Pp},\sfrac{\Lb}{\Pp})$ is either splendid or shallow.

Define a decomposition $\sfrac{\Tt}{\Pp}=(\sfrac{T}{\Pp},\sfrac{\eta}{\Pp})$ of $G$ as follows: for every $u\in V(G)$, we set $\sfrac{\eta}{\Pp}(u)$ to be the top vertex of the unique factor of $\Pp$ that contains $\eta(u)$. Let $\Lambda|_{\Pp}$ be the restriction of the tagging $\Lambda$ to the nodes of $\tp(\Pp)$, which is then a tagging of $\sfrac{\Tt}{\Pp}$. The following claim then follows from definitions in a straightforward manner; here, $\overline{\Lambda|_{\Pp}}$ is the labelling induced by tagging $\Lambda|_{\Pp}$ in $\sfrac{\Tt}{\Pp}$.

\begin{claim}\label{cl:commutes}
It holds that $\overline{\Lambda|_{\Pp}}=\sfrac{\Lb}{\Pp}$.
\end{claim}

It now suffices to check that the decomposition $\sfrac{\Tt}{\Pp}$ posesses all the properties asserted in the lemma statement.
Note that for every node $x$ of $\tp(\Pp)$ we have $\ang{V}{x,\sfrac{\Tt}{\Pp}}=\ang{V}{x,\Tt}$, hence we use the shorthand $\ang{V}{x}$ for both.

\begin{claim}\label{cl:diversity}
$\sfrac{\Tt}{\Pp}$ has diversity at most $k$.
\end{claim}
\begin{proof}
Since for every node $x\in \tp(\Pp)$ we have $\ang{V}{x,\sfrac{\Tt}{\Pp}}=\ang{V}{x,\Tt}=\ang{V}{x}$, also the relation $\sim_x$ over $\ang{V}{x}$ computed in $\Tt$ is the same as computed in $\sfrac{\Tt}{\Pp}$. As $\Tt$ has diversity at most $k$, this relation has always at most $k$ equivalence classes, which implies that $\sfrac{\Tt}{\Pp}$ also has diversity at most $k$.
\cqed\end{proof}

\begin{claim}
$\sfrac{\Tt}{\Pp}$ is $\Dd_{i-1}$-governed.
\end{claim}
\begin{proof}
Consider any node $x$ of $\sfrac{T}{\Pp}$. 
Then $x=\tp(F)$ for some factor $F$ of $\Pp$.
Let $\eta_F\colon \ang{V}{x}\to V(F)$ be defined as follows: for $u\in \ang{V}{x}$ we put
\begin{itemize}
\item $\eta_F(u)=\eta(u)$, if $\eta(u)$ belongs to $F$; and
\item $\eta_F(u)$ to be the least ancestor of $\eta(u)$ that belongs to $F$, otherwise.
\end{itemize}
Let $H=\ang{G}{x,\sfrac{\Tt}{\Pp}}$ that is, the vertex set of $H$ is $\ang{V}{x}$ and the edge set of $H$ consist of all edges $uv$ of $G$ such that the least common ancestor of $\eta(u)$ and $\eta(v)$ in $T$ belongs to $F$. Then
 $(F,\eta_F)$ is a decomposition of $H$. A similar argument as in Claim~\ref{cl:diversity} shows that the diversity of $(F,\eta_F)$ is at most $k$. Also, $(F,\eta_F)$ as a decomposition of $H$ is $\Cc$-governed, because 
 $$\ang{H}{y,(F,\eta_F)}=\ang{G}{y,\Tt}\qquad\textrm{for each node }y\textrm{ of }F.$$
Moreover, if $\Lambda|_F$ is the restriction of the tagging $\Lambda$ to the nodes of $F$ and $\Lb|_{F}$ is the restriction of the labelling $\Lb$ to the edges of $F$, then it is easy to see that $\overline{\Lambda|_F}=\Lb|_{F}$. Since $(F,\Lb|_{F})\in \Hh_{i-1}$, we infer that $(F,\eta_F)$ together with tagging $\Lambda|_F$ witnesses that $H\in \Dd_{i-1}$. As $x$ was taken arbitrarily, we conclude that $\sfrac{\Tt}{\Pp}$ is $\Dd_{i-1}$-governed.
\cqed\end{proof}

\begin{claim}
$\sfrac{\Tt}{\Pp}$ is either splendid or shallow.
\end{claim}
\begin{proof}
Recall that $(\sfrac{T}{\Pp},\sfrac{\Lb}{\Pp})$, as an $\Ff$-labelled tree, is either splendid or shallow. In the latter case we can immediately conclude, because then the decomposition $\sfrac{\Tt}{\Pp}$ has depth at most $2$. In the former case, by Claim~\ref{cl:commutes} we infer that $(\sfrac{T}{\Pp},\overline{\Lambda|_{\Pp}})$ is splendid. This means that the tagging $\Lambda|_{\Pp}$ witnesses that $\sfrac{\Tt}{\Pp}$ is splendid.
\cqed\end{proof}

The above three claims verify that $\sfrac{\Tt}{\Pp}$ has the required properties and we are done.
\end{proof}

\section{Treating shallow and splendid decompositions}\label{sec:splendid}

\subsection{Shallow decompositions}

For shallow decompositions we may use a direct product argument.

\begin{lemma}\label{lem:depth-2}
Let $\Cc$ be a class of graphs that is $\chi$-bounded by a non-decreasing function~$g$.
Then the class of graphs that admit $\Cc$-governed shallow decompositions is $\chi$-bounded by the function $g^2$.
\end{lemma}
\begin{proof}
Let $G$ be a graph that admits a decomposition $(T,\eta)$ that is $\Cc$-governed and shallow. Let $x$ be the root of $T$; then $\ang{V}{x}=V(G)$. By assumption $\ang{G}{x}\in \Cc$ and $\ang{G}{y}\in \Cc$ for every child $y$ of $x$. Hence, we can find a proper coloring $\phi_x$ of $\ang{G}{x}$ with at most $g(\omega(\ang{G}{x}))\leq g(\omega(G))$ colors, and, for every child $y$ of $x$, a proper coloring $\phi_y$ of $\ang{G}{y}$ with at most $g(\omega(\ang{G}{y}))\leq g(\omega(G))$ colors. Then define a coloring $\phi$ of $G$ as follows: for any $u\in V(G)$, we let
$$\phi(u) = (\phi_x(u),\phi_{\eta(u)}(u)).$$
Clearly, $\phi$ uses at most $g(\omega(G))^2$ colors, hence it suffices to verify that $\phi$ is a proper coloring of $G$.

Consider any edge $e=uv$ of $G$. If $\eta(u)=\eta(v)$, then also $\eta(e)=\eta(u)=\eta(v)$ and the edge $e$ is present in the graph $\ang{G}{\eta(e)}=\ang{G}{\eta(u)}=\ang{G}{\eta(v)}$. As $\phi_{\eta(e)}$ is a proper coloring of $\ang{G}{\eta(e)}$, the second coordinates of $\phi(u)$ and $\phi(v)$ differ. Otherwise, if $\eta(u)\neq \eta(v)$, then $\eta(e)=x$, because $T$ has depth at most $2$. Then the edge $e$ is present in the graph $\ang{G}{x}$ and, since $\phi_x$ is a proper coloring of $\ang{G}{x}$, we conclude that the first coordinates of $\phi(u)$ and $\phi(v)$ differ. In both cases we have $\phi(u)\neq \phi(v)$, hence $\phi$ is a proper coloring of $G$.
\end{proof}

\subsection{Splendid decompositions}

For splendid decompositions, the main tool will be a result of Chudnovsky et al.~\cite{ChudnovskyPST13}, which treats of the closure of classes of graphs under the operation of substitution, defined as follows. Suppose that $G$ is a graph and $\alpha$ is a mapping that associates with each vertex $u$ of $G$ some graph $\alpha(u)$. Then we define the graph $G[\alpha]$ as follows:
\begin{itemize}
    \item The vertex set of $G[\alpha]$ consists of pairs of the form $(u,v)$, where $u$ is a vertex of $G$ and $v$ is a vertex of $\alpha(u)$.
    \item Two vertices $(u,v)$ and $(u',v')$ are adjacent in $G[\alpha]$ if either $u\neq u'$ and $uu'$ is an edge in $G$, or $u=u'$ and $vv'$ is an edge in $\alpha(G)$.
\end{itemize}
Informally speaking, $G[\alpha]$ is obtained from $G$ by replacing every vertex $u$ with the graph $\alpha(u)$, and putting a complete join between graphs $\alpha(u)$ and $\alpha(u')$ whenever $uu'$ was an edge of $G$.

For a graph class $\Cc$, let $\Cc^\star$ be the closure of $\Cc$ under the substitution operation. That is, $\Cc^\star$ is the smallest class that contains $\Cc$ and whenever $G\in \Cc^\star$ and the image of $\alpha$ is contained in $\Cc^\star$, then $G[\alpha]\in \Cc^\star$ as well. The following claim links the operation of substitution with our terminology and is easy to verify; we leave it without a proof, as we will not use it later on.

\begin{lemma}\label{lem:sub-k}
For every hereditary graph class $\Cc$ closed under substituting vertices with edgeless graphs, the class $\Cc^\star$ is exactly the class of all graphs that admit a $\Cc$-governed decomposition of diversity $1$.
\end{lemma}

Chudnovsky et al.~\cite{ChudnovskyPST13} proved the following.

\begin{theorem}[Theorem 2.3 of~\cite{ChudnovskyPST13}]\label{thm:substitution}
If a hereditary graph class $\Cc$ is polynomially $\chi$-bounded, then so is $\Cc^{\star}$.
\end{theorem}

Thus, by Lemma~\ref{lem:sub-k}, Chudnovsky et al. in fact essentially proved Theorem~\ref{thm:main-technical} for $k=1$. We will now use this result to prove Theorem~\ref{thm:main-technical} in the case when the provided decomposition is splendid.

\setcounter{claim}{0}

\begin{lemma}\label{lem:splendid}
Let $\Cc$ be a hereditary class of graphs that is polynomially $\chi$-bounded. Then the class of graphs that admit splendid $\Cc$-governed decompositions of diversity at most $k$ is also polynomially $\chi$-bounded.
\end{lemma}
\begin{proof}
First, we need to understand forward Ramsey sets in the semigroup $\Ff$.

\begin{claim}\label{cl:compatible}
Suppose $A\subseteq \Ff$ is forward Ramsey. Then there exists an equivalence relation $\cong$ on $[k]$ such for every $f\in A$ and all $i,j\in [k]$ satisfying $i\cong j$, we have $f(i)=f(j)$ and $f(i)\cong i$.
\end{claim}
\begin{proof}
Fix any $e\in A$ and define the equivalence relation $\cong$ as follows: for $i,j\in [k]$, we have $i\cong j$ if and only if $e(i)=e(j)$. To verify that $\cong$ defined in this way has the required property, consider any $f\in A$ and any $i,j\in [k]$ such that $i\cong j$, that is, $e(i)=e(j)$. Since $f\cdot e=f$ due to $A$ being forward Ramsey, we have 
$$f(i)=(f\cdot e)(i)=f(e(i))=f(e(j))=(f\cdot e)(j)=f(j).$$
Similarly we have $e\cdot f=e$, hence
$$e(i)=(e\cdot f)(i)=e(f(i)),$$
implying that $i\cong f(i)$.
\cqed\end{proof}

We proceed to the main proof. We may assume that $\Cc$ contains all edgeless graphs, because adding all such graphs to $\Cc$ does not spoil the assumption that $\Cc$ is hereditary and polynomially $\chi$-bounded.


Consider any graph $G$ that admits a splendid $\Cc$-governed decomposition $\Tt=(T,\eta)$ of diversity at most $k$. Let $\Lambda=(\lambda^x)_{x\in V(T)}$ be a tagging of $\Tt$ that witnesses that $\Tt$ is splendid, that is, the set $\{\Lb(e)\colon E(T)\}$ is forward Ramsey. Let $\cong$ be the equivalence relation on $[k]$ provided by Claim~\ref{cl:compatible} for this set.

\begin{claim}\label{cl:partition-classes}
For every $u\in V(G)$ and $x,y\in V(T)$ satisfying $u\in \ang{V}{x}$ and $u\in \ang{V}{y}$, we have $\lambda^x(u)\cong\lambda^y(u)$.
\end{claim}
\begin{proof}
Note that the nodes $z\in V(T)$ satisfying $u\in \ang{V}{z}$ are exactly the ancestors of $\eta(u)$. As these ancestors induce a path in $T$, by the transitivity of $\cong$ it suffices to prove the claim in the case when $x$ and $y$ are adjacent in $T$, say $x$ is the parent of $y$.

Let $f=\Lb(xy)\in \Ff$. By the definition of $\Lb$ we have $f(\lambda^y(u))=\lambda^x(u)$, whereas by Claim~\ref{cl:compatible} we have $f(\lambda^y(u))\cong \lambda^y(u)$. The claim follows.
\cqed\end{proof}

By Claim~\ref{cl:partition-classes}, with every vertex $u\in V(G)$ we can associate an equivalence class $\tau(u)$ of $\cong$ with the following property:
$$\lambda^x(u)\in \tau(u)\qquad\textrm{for every node } x\textrm{ such that }u\in \ang{V}{x}.$$
For a class $\kappa$ of $\cong$, let
$$G^\kappa = G[\tau^{-1}(\kappa)]\qquad\textrm{and}\qquad \eta^{\kappa}=\eta|_{\tau^{-1}(\kappa)}.$$
Note that $(T,\eta^{\kappa})$ is a decomposition of $G^\kappa$ of diversity at most $|\kappa|\leq k$. Furthermore, tagging $\Lambda$ restricted to the vertices of $G^\kappa$ witnesses that $(T,\eta^{\kappa})$ is splendid. Finally, since $\Cc$ is hereditary, we conclude that $(T,\eta^{\kappa})$ is $\Cc$-governed.

We now argue that from now on we can focus on the case when $\cong$ has only one equivalence class. Indeed, suppose that under this assumption we are able to prove that $\chi(G)\leq h(\omega(G))$, for some non-decreasing polynomial $h$. Then in the general case, we may apply this reasoning to the graph $G^{\kappa}$, for every equivalence class $\kappa$ of $\cong$, thus showing that $G^\kappa$ can be properly colored with $h(\omega(G^\kappa))\leq h(\omega(G))$ colors. It now suffices to take the union of these colorings, using a different set of $h(\omega(G))$ colors for each of them, in order to see that $G$ can be properly colored with $k\cdot h(\omega(G))$ colors.

Hence, from now on we assume that $\cong$ has only one equivalence class, which is equivalent to the following assertion:
\begin{equation}\label{eq:cnst}
\textrm{for each }e\in E(T)\textrm{, the mapping }\Lb(e)\textrm{ is a constant function.}
\end{equation}

For a node $x$ of $T$, the {\em{depth}} of $x$ is the number of nodes on the path from the root of $T$ to~$x$. We now partition the edge set of $G$ into $E^0$ and $E^1$ as follows:
$$E^0=\{e\in E(T)\colon \eta(e)\textrm{ is at even depth}\},\qquad 
  E^1=\{e\in E(T)\colon \eta(e)\textrm{ is at odd depth}\}.$$
Also, we define subgraphs $G_0$ and $G_1$ of $G$ as follows:
$$G^0=(V(G),E^0),\qquad 
  G^1=(V(G),E^1).$$
The intuition now is that we would like to color each of the graphs $G^0,G^1$ separately and superpose the obtained two colorings. The following claim, which is the core argument of the proof, reduces this task to a direct application of Theorem~\ref{thm:substitution}.

\newcommand{\dwn}[2]{#1_{#2}}

\begin{claim}\label{cl:in-Cstar}
$G^0,G^1\in \Cc^\star$.
\end{claim}
\begin{proof}
We prove the statement for $G^0$, as the proof for $G^1$ is the same.
By a bottom-up induction on $T$ we argue the following statement: for every node $x$, the graph $\dwn{G^0}{x}:=G^0[\ang{V}{x}]$ belongs to $\Cc^\star$. Then the claim follows from applying this statement to the root of $T$.

Consider first the case when the depth of $x$ is odd. Then $\dwn{G^0}{x}$ consists of the disjoint union of: graphs $\dwn{G^0}{y}$ for $y$ ranging over children of $x$; and vertices $v$ satisfying $\eta(v)=x$, which are isolated in $\dwn{G^0}{x}$. By the induction assumption, graphs $\dwn{G^0}{y}$ belong to $\Cc^\star$. Thus, $\dwn{G^0}{x}$ can be obtained from an edgeless graph by substituting some of its vertices with graphs from $\Cc^\star$. Since we assumed that $\Cc$ contains all edgeless graphs, we conclude that $\dwn{G^0}{x}\in \Cc^\star$ as well.

Consider now the case when the depth of $x$ is even. Take any grandchild $z$ of $x$; that is, $z$ is a child of some child $y$ of $x$. We claim that
\begin{equation}\label{eq:moduleG0x}
N^{\dwn{G^0}{x}}(u)\setminus \ang{V}{z} = N^{\dwn{G^0}{x}}(v)\setminus \ang{V}{z} \quad\textrm{for all } u,v\in \ang{V}{z}.
\end{equation}
To see this, observe the following. 
First, by the construction of $\dwn{G^0}{x}$, in this graph $u$ and $v$ have no neighbors in $\ang{V}{y}\setminus \ang{V}{z}$.
Next, by~\eqref{eq:cnst} applied to the edge $yz$, we have 
$$\lambda^y(u)=\Lb(yz)(\lambda^z(u))=\Lb(yz)(\lambda^z(v))=\lambda^y(v).$$
By the definition of tagging, this implies that $u\sim_y v$, or equivalently
\begin{equation}\label{eq:smashed}
    N^G(u)\setminus \ang{V}{y}=N^G(v)\setminus \ang{V}{y}.
\end{equation}
This, in turn, entails
$$N^{\dwn{G^0}{x}}(u)\setminus \ang{V}{y}=N^{\dwn{G^0}{x}}(v)\setminus \ang{V}{y},$$
which together with $u$ and $v$ having no neighbors in $\ang{V}{y}\setminus \ang{V}{z}$ establishes~\eqref{eq:moduleG0x}.
Observe that by the definition of $\ang{G}{x}$, the set $\ang{V}{y}$ is independent in this graph. Therefore, by~\eqref{eq:smashed}, we can also conclude that 
\begin{equation}\label{eq:moduleGx}
N^{\ang{G}{x}}(u)\setminus \ang{V}{z} = N^{\ang{G}{x}}(v)\setminus \ang{V}{z} \quad\textrm{for all } u,v\in \ang{V}{z}.
\end{equation}

Let $H$ be the induced subgraph of $\ang{G}{x}$ obtained by removing, for every grandchild $z$ of $x$ with $\ang{V}{z}\neq \emptyset$, all but one vertex of $\ang{V}{z}$. Note here that, by~\eqref{eq:moduleGx}, $H$ does not depend on which vertex of $\ang{V}{z}$ is chosen, up to isomorphism. Define the following mapping on vertices of $H$:
\begin{itemize}
    \item if $u\in V(H)$ is such that $\{u\}=V(H)\cap \ang{V}{z}$ for some grandchild $z$ of $x$, then $\alpha(u)=\dwn{G^0}{z}$; and
    \item otherwise, $\alpha(u)$ is the one-vertex graph consisting only of $u$.
\end{itemize}
Then by~\eqref{eq:moduleG0x} we conclude that
\begin{equation}\label{eq:subst}
    H[\alpha]\textrm{ and }\dwn{G^0}{x}\textrm{ are isomorphic.}
\end{equation}
Since $(T,\eta)$ is $\Cc$-governed, we have $\ang{G}{x}\in \Cc$, which entails $H\in \Cc$ because $\Cc$ is hereditary. Further, all the graphs in the image of $\alpha$ belong to $\Cc^\star$ by the induction assumption. Then by~\eqref{eq:subst} we may conclude that $\dwn{G^0}{x}\in \Cc$, as claimed.
\cqed\end{proof}


Since $\Cc$ is polynomially $\chi$-bounded, by Theorem~\ref{thm:substitution} we conclude that $\Cc^\star$ is polynomially $\chi$-bounded as well, say by a nondecreasing polynomial $g(\cdot)$. Therefore, by Claim~\ref{cl:in-Cstar} we conclude that there are proper colorings $\phi^0$ and $\phi^1$ of $G^0$ and $G^1$, respectively, where $\phi^0$ uses at most $g(\omega(G^0))\leq g(\omega(G))$ colors and $\phi^1$ uses at most $g(\omega(G^1))\leq g(\omega(G))$ colors. Since every edge of $G$ belongs either to $G^0$ or to $G^1$, we conclude that $\phi$ defined as $\phi(u)=(\phi^0(u),\phi^1(u))$ is a proper coloring of $G$ with $g(\omega(G))^2$ colors. 
Keeping in mind the multiplicative factor of $k$ incurred by the reduction to the case when $\cong$ has one equivalence class, we conclude that the class of graphs that admit splendid $\Cc$-governed decompositions of diversity at most $k$ is $\chi$-bounded by the polynomial $t\mapsto k\cdot g(t)^2$.
\end{proof}

\subsection{Concluding the proof}

We may now formally conclude the proof of our main result.

\begin{proof}[Proof of Theorem~\ref{thm:main-technical}]
Let $\Cc=\Dd_0\subseteq \Dd_1\subseteq\ldots\subseteq \Dd_p=\Cc^\clw$ be the sequence of hereditary classes provided by Lemma~\ref{lem:hierarchy}. By induction on $i$ we prove that each class $\Dd_i$ is polynomially $\chi$-bounded, with the base case provided by the assumption that $\Cc=\Dd_0$ is polynomially $\chi$-bounded.
By Lemma~\ref{lem:hierarchy}, the class $\Dd_i$ is contained in the union of two graph classes: the class of graphs admitting shallow $\Dd_{i-1}$-governed decompositions of diversity at most $k$, and the class of graphs admitting splendid $\Dd_{i-1}$-governed decompositions of diversity at most $k$. The induction assumption and Lemma~\ref{lem:depth-2} imply that the former class is polynomially $\chi$-bounded, and similarly from Lemma~\ref{lem:splendid} we conclude that the latter class is polynomially $\chi$-bounded. As the union of two polynomially $\chi$-bounded classes is polynomially $\chi$-bounded, we infer that $\Dd_i$ is polynomially $\chi$-bounded, which proves the induction step. Thus $\Dd_p=\Cc^\clw$ is polynomially $\chi$-bounded and we are done.
\end{proof}

\section{Conclusions}\label{sec:concl}

\subsection{Discussion of the asymptotics}\label{sec:asymptotics}

As stated, our main result only asserts that if the class $\Cc$ is polynomially $\chi$-bounded then $\Cc^\clw$ is polynomially $\chi$-bounded as well. However, we find it instructive to trace the explosion of $\chi$-bounding polynomials throughout the proof.

The proof of Theorem~\ref{thm:substitution} in~\cite{ChudnovskyPST13} shows that if $\Cc$ is $\chi$-bounded by the polynomial $t\mapsto t^A$, for a positive integer $A$, then $\Cc^\star$ is $\chi$-bounded by the polynomial $t\mapsto t^{3A+11}$. Thus, Lemma~\ref{lem:splendid} incurs a blow-up of the $\chi$-bounding polynomial from $t\mapsto t^A$ (for $\Cc)$ to $t\mapsto k\cdot t^{6A+22}$ (for the class in the conclusion of the lemma). Similarly, Lemma~\ref{lem:depth-2} incurs a blow-up of the $\chi$-bounding function from $t\mapsto t^A$ to $t\mapsto t^{2A}$. Thus, in the proof of Theorem~\ref{thm:main-technical} we may conclude that if $\Dd_{i-1}$ is $\chi$-bounded by the polynomial $t\mapsto t^A$, then $\Dd_i$ is $\chi$-bounded by $t\mapsto k\cdot t^{6A+22}$. Along the induction this blow-up occurs $2^{\Oh(k\log k)}$ times, hence we can reach the following conclusion.

\begin{corollary}
If $\Cc$ is $\chi$-bounded by the polynomial $t\mapsto t^A$ for some positive integer $A$, then $\Cc^{\clw}$ is $\chi$-bounded by a polynomial of the form $t\mapsto t^{A\cdot 2^{2^{\Oh(k\log k)}}}$.
\end{corollary}

For the case of graphs of cliquewidth at most $k$, which were the primary motivation of this work, this gives an upper bound of $t\mapsto t^{2^{2^{\Oh(k\log k)}}}$ on the $\chi$-bounding polynomial, because, by Lemma~\ref{lem:cwk}, the base class $\Cc$ is formed by bipartite graphs, which are perfect (i.e., $A=1$). 

On the other hand, in the proof of Proposition~2.4 of~\cite{ChudnovskyPST13}, Chudnovsky et al. argued that if $F$ is a triangle-free graph with fractional chromatic number larger than $2^d$, and we inductively define $F_1=F$ and $F_{i+1}$ as $F$ with every vertex substituted with $F_i$, then
\begin{itemize}
    \item $\omega(F_i)=2^i$; and
    \item $\chi(F_i)>2^{id}$.
\end{itemize}
It is easy to see that each graph $F_i$ has cliquewidth at most $|V(F)|$. As a graph $F$ with properties as above can be chosen so that it has $\Oh(2^{2d})$ vertices, see the discussion in~\cite{ChudnovskyPST13}, this proves the following corollary.

\begin{lemma}\label{lem:grows-with-k}
For every $k\in \N$, if $g$ is a $\chi$-bounding polynomial for the class of graphs of cliquewidth at most $k$, then the degree of $g$ is at least $\Omega(\log k)$.
\end{lemma}

While Lemma~\ref{lem:grows-with-k} ensures us that the degree of the $\chi$-bounding polynomial for graphs of cliquewidth at most $k$ needs to grow with $k$, there is still a significant gap between the upper bound --- two-fold exponential in $k$ --- and the lower bound --- logarithmic in $k$.

\subsection{$1$-joins and graphs with excluded vertex-minor}\label{sec:wheel-free}

For two graphs $G_1$ and $G_2$ with distinguished vertices $u_1$ and $u_2$, respectively, their {\em{$1$-join}} is the graph obtained from their disjoint union by removing $u_1$ and $u_2$ and making every former neighbor of $u_1$ adjacent to every former neighbor of $u_2$.
If $\Cc$ is a class of graphs, then by $\Cc^{\&}$ we denote the closure of $\Cc$ under repeated application of the $1$-join operation.
As observed by Dvo\v{r}\'ak and Kr\'al' in~\cite{DvorakK12}, such closure under $1$-joins corresponds to having a decomposition of diversity at most $2$ in the following sense.

\begin{lemma}[\cite{DvorakK12}]\label{lem:one-join-decomp}
If $\Cc$ is a hereditary class of graphs closed under adding isolated vertices and substituting vertices with edgeless graphs, then every graph from $\Cc^{\&}$ admits a $\Cc$-governed decomposition of diversity at most $2$.
\end{lemma}

We remark that Dvo\v{r}\'ak and Kr\'al' speak about $\Cc$-bounded decompositions of rank at most~$1$, which in our terminology are $\Cc$-governed decompositions of diversity at most $2$.
We refrain from giving a formal proof of Lemma~\ref{lem:one-join-decomp} here, because a reasoning leading to it can be found essentially verbatim in the proofs of Theorems~2 and~3 of~\cite{DvorakK12}.

The main result of Dvo\v{r}\'ak and Kr\'al' states (in our terminology) that if a hereditary graph class $\Cc$ is $\chi$-bounded, then for every fixed $k\in \N$ the class of graphs that admit a $\Cc$-governed decomposition of diversity at most $k$ is also $\chi$-bounded. Thus, by combining Lemma~\ref{lem:one-join-decomp} with the main result of~\cite{DvorakK12} on one side, and with our Theorem~\ref{thm:main-technical} on the other side, both for $k=2$, we obtain the following corollary.

\begin{corollary}\label{cor:one-joins}
If $\Cc$ is a hereditary graph class that is $\chi$-bounded, then $\Cc^{\&}$ is also $\chi$-bounded. If $\Cc$ is moreover polynomially $\chi$-bounded, then so is $\Cc^{\&}$ as well.
\end{corollary}

The first part of Corollary~\ref{cor:one-joins} was explicitely proved by Dvo\v{r}\'ak and Kr\'al' in~\cite{DvorakK12}, while the second was proved by Kim et al.~\cite{KiKOS19} using an argument tailored to $1$-joins; hence, this is not a new result. 

Corollary~\ref{cor:one-joins} is particularly useful when working with classes of graphs defined by forbidding vertex-minors. Recall here that $H$ is a {\em{vertex-minor}} of a graph $G$ if $H$ can be obtained from $G$ by a sequence of operations: vertex deletion and {\em{local complementation}}, which amounts to swapping the adjacency and non-adjacency relations in the neighborhood of a vertex.
It turns out that for some simple graphs $H$ we have decomposition results of the following type: every $H$-vertex-minor-free graph either can be obtained as a $1$-join of two smaller graphs, or belongs to a simpler class of graphs. For instance, for the case when $H$ is the wheel $W_5$, Geelen~\cite{geelen-thesis} gave such a decomposition theorem where the simpler class of graphs comprises of all circle graphs and a few graphs on at most $8$ vertices. By combining this result with the first part of the statement of Corollary~\ref{cor:one-joins} and the fact that circle graphs are $\chi$-bounded~\cite{Gyarfas85,Gyarfas85corr}, Dvo\v{r}\'ak and Kr\'al' concluded in~\cite{DvorakK12} that the class of $W_5$-vertex-minor-free graphs is $\chi$-bounded. Note here that by using the second part of statement of Corollary~\ref{cor:one-joins} and the recent result of Davies and McCarty that circle graphs are quadratically $\chi$-bounded~\cite{DaviesMcC19}, we can in the same manner argue that $W_5$-vertex-minor-free graphs are polynomially $\chi$-bounded. In a similar manner --- by using the second part of statement of Corollary~\ref{cor:one-joins} together with a suitable decomposition result using $1$-joins --- Kim et al.~\cite{KiKOS19} proved that the class of $C_\ell$-vertex-minor-free graphs is polynomially $\chi$-bounded, for every $\ell\geq 3$.

It seems naive to assume that a suitable structural result involving just $1$-joins would hold in the class of $H$-vertex-minor-free graphs, for every fixed graph $H$. However, it is believed that graphs with a fixed excluded vertex-minor should admit a decomposition theorem of a form analogous to the decomposition theorem for proper minor-closed classes of Robertson and Seymour~\cite{RobertsonS03a}, where instead of $1$-joins we use cuts of larger diversity. That is, if $\Dd$ is a class closed under taking vertex-minors that does not contain all graphs, then there exists $k\in \N$ such that every graph from $\Dd$ admits a $\Cc$-governed decomposition of diversity at most $k$, where $\Cc$ is some class that generalizes circle graphs in a similar way as nearly embeddable graphs generalize planar graphs. 
Thus, the work of Dvo\v{r}\'ak and Kr\'al'~\cite{DvorakK12} suggests a possible route towards proving the following conjecture attributed to Geelen:

\begin{conjecture}\label{conj:vertex-minor}
For every graph $H$, the class of graphs that exclude $H$ as a vertex-minor is $\chi$-bounded.
\end{conjecture}

More precisely, the main result  of Dvo\v{r}\'ak and Kr\'al'~\cite{DvorakK12} reduces proving $\chi$-boundedness of the class $\Dd$ to proving $\chi$-boundedness of $\Cc$, which may follow from a lift of the reasoning for circle graphs. Similarly, our Theorem~\ref{thm:main-technical} reduces proving polynomial $\chi$-boundedness of $\Dd$ to proving polynomial $\chi$-boundedness of~$\Cc$. Together with the quadratic $\chi$-boundedness of circle graphs~\cite{DaviesMcC19}, this suggests that in Conjecture~\ref{conj:vertex-minor} one might expect even polynomial $\chi$-boundedness. This question was also posed by Kim et al., see~\cite[Question 1.4]{KiKOS19}. 

Very recently, Geelen et al.~\cite{geelen2019grid} proved that for any circle graph $H$, the class of graphs that exclude $H$ as a vertex-minor has bounded rankwidth. Together with our result, this implies that such a class is polynomially $\chi$-bounded.


\section*{Acknowledgments} 
We would like to express our gratitude to Rose McCarty for communicating the problem to us and for bringing the works of Chudnovsky et al.~\cite{ChudnovskyPST13} and of Kim et al.~\cite{KiKOS19} to our attention. We thank the anonymous reviewers for their efficiency and their help in improving the presentation of this paper.

\bibliographystyle{amsplain}
\newcommand{\etalchar}[1]{$^{#1}$}


\begin{aicauthors}
\begin{authorinfo}[mb]
  Marthe Bonamy\\
  CNRS, LaBRI, Universit\'e de Bordeaux\\
  Bordeaux, France\\
  marthe\imagedot{}bonamy\imageat{}u-bordeaux\imagedot{}fr \\
  \url{https://www.labri.fr/perso/mbonamy/}
\end{authorinfo}
\begin{authorinfo}[mp]
  Micha\l \ Pilipczuk\\
  Institute of Informatics, University of Warsaw\\ 
  Warsaw, Poland\\
  michal\imagedot{}pilipczuk\imageat{}mimuw\imagedot{}edu\imagedot{}pl \\
  \url{https://www.mimuw.edu.pl/~mp248287/}
\end{authorinfo}
\end{aicauthors}

\end{document}